\documentclass[a4paper,12pt]{article}

\usepackage{amsmath,amssymb,amsthm}
\usepackage{mathtools}
\usepackage[margin=24mm]{geometry}
\usepackage[doublespacing]{setspace}
\usepackage{enumerate}
\usepackage{floatrow}
\usepackage[round]{natbib}

\newcommand{\keywords}[1]{\par\noindent\textbf{Keywords:}%
 #1\par}
\newcommand{\jel}[1]{\par\noindent\textbf{JEL classification codes:}%
 #1\par}
\newcommand{\msc}[1]{\par\noindent\textbf{Mathematics subject classification 2020:}%
 #1\par}
\makeatletter
\renewcommand{\p@enumii}{}
\makeatother
\makeatletter
\newcommand{\subproof}[1]{\if@inlabel\leavevmode\else\fi\par\textit{Proof of #1.}}
\makeatother
\newcommand{\rnd}[1]{\left(#1\right)}
\newcommand{\crly}[1]{\left\{#1\right\}}

\newcommand{\arrow}{\overset{\sim}{\leftrightarrow}}
\newcommand{\set}[2]{\left\{#1 \mid #2\right\}}

\newcommand{\Map}[2]{\operatorname{Map}\rnd{#1,#2}}
\newcommand{\dom}{\operatorname{dom}}
\newcommand{\cod}{\operatorname{cod}}
\newcommand{\graph}{\operatorname{graph}}
\newcommand{\id}{\operatorname{id}}
\newcommand{\ob}[1]{\operatorname{ob}\rnd{#1}}
\newcommand{\mor}[1]{\operatorname{mor}\rnd{#1}}
\newcommand{\op}[1]{#1^{\mathrm{op}}}
\newcommand{\Set}{\mathbf{Set}}
\newcommand{\Rel}{\mathbf{Rel}}
\newcommand{\Gam}{\mathbf{Gam}}
\newcommand{\NE}{\mathrm{NE}}
\newcommand{\PE}{\mathrm{PE}}
\renewcommand{\hom}[3]{\operatorname{hom}_{#1}\rnd{#2,#3}}
\newenvironment{proofsketch}{\begin{proof}[Proof sketch]}{\end{proof}}
\newtheorem{proposition}{Proposition}

\newtheorem{lemma}{Lemma}

\theoremstyle{definition}

\newtheorem{condition}{Condition}

\theoremstyle{remark}

\title
{
 Category of Strategic Games
 and
 Presheaf Corresponding to Solution Concepts or Welfare Criteria
}
\author
{
 Tomohiko Kawamori\thanks
 {
  Faculty of Economics, Meijo University,
  1-501 Shiogamaguchi, Tempaku-ku, Nagoya 468-8502, Japan.
  {\tt kawamori@meijo-u.ac.jp}
 }
}
\date
{
}

\begin{document}

\maketitle

\abstract
{
 We
 define a category of strategic games in which
 a morphism is a pair of a map between sets of players and a map between sets of strategy profiles with specific properties
 and
 show that
 this category is well-defined.
 Three cases are considered for the maps between sets of strategy profiles:
 they may be order-preserving, order-reflecting or order-embedding with respect to each player's preference relation.
 We
 define a presheaf on the category of games valued in a category of sets that sends each strategic game to a set of strategy profiles
 and
 present conditions for this presheaf to be well-defined.
 Two cases are considered for the morphisms in the category of sets:
 they may be relations or maps.
 We
 define a presheaf that sends each strategic game to the set of Nash equilibria (resp. Pareto efficient strategy profiles)
 and
 show that
 this presheaf is well-defined
 if and only if
 the maps between sets of strategy profiles are order-reflecting or order-embedding
 (resp. order-embedding),
 and
 the morphisms in the category of sets are relations.
}


\jel
{
 C65;
 C72
}

\msc
{
 18B99;
 91A06;
 91A10
}

\newpage
\section{Introduction}\label{sec:introduction}

This paper seeks to reframe the core concepts of game theory within the language of category theory.
Category theory provides a general framework for analyzing relationships between mathematical structures.\footnote
{
 See \cite{Awodey2010}.
}
It is often regarded as a universal language of mathematics,
much as set theory has long been.
Its applications extend not only across a wide range of fields of mathematics but also into areas such as computer science and philosophy beyond mathematics.
A category-theoretic framework for game theory
would
clarify its connections with other fields
or
enable the use of tools developed in those fields.
However,
category-theoretic studies of game theory remain limited.
Even foundational questions,
such as how to define the category of strategic games,
have received several answers,
yet
have not yielded definitive ones.
In response to this gap,
this paper proposes another category-theoretic framework for game theory.
By doing so,
this paper aims to further the possibility of studying game theory by analogy with other mathematical fields.

This paper
defines a category of strategic games
and
shows that
this category is well-defined.
In this category,
an object is a strategic game,
while a morphism between two strategic games is a pair of a map between the sets of players and a map between the sets of strategy profiles with specific properties.
Three cases are considered for the maps between the sets of strategy profiles:
they may be order-preserving, order-reflecting or order-embedding with respect to each player's preference relation.
This definition is a natural and simple variant of that of the category of preordered sets.

This paper
defines presheaves corresponding to the Nash equilibrium concept and the Pareto efficiency criterion
and
examines conditions under which these presheaves are well-defined.
To do this,
we first define a presheaf on the category of strategic games valued in a category of sets that sends each strategic game to a set of strategy profiles.
Two cases are considered for the morphisms of the category of sets:
they may be relations or maps.
We present an equivalent condition for this presheaf to be well-defined.
As a particular instance of this presheaf,
we define a presheaf corresponding to the Nash equilibrium concept (resp. the Pareto efficiency criterion)
that sends each strategic game to the set of Nash equilibria (resp. Pareto efficient strategy profiles).
We show that
when the morphisms of the category of sets are relations,
this presheaf corresponding to the Nash equilibrium concept (resp. the Pareto efficiency criterion) is well-defined
if and only if the maps between the set of strategy profiles are order-reflecting or order-embedding (resp. order-embedding).
In contrast,
we show that
when the morphisms of the category of sets are maps,
this presheaf is not well-defined.

Our results are characterized by four notable features.
(1)
The restrictions imposed on morphisms between strategic games are minimal,
and
these morphisms are natural variants of those in the category of preordered sets:
maps between sets of strategy profiles are merely required to be order-preserving, order-reflecting or order-embedding,
without any further constraints such as surjectivity.
(2)
Owing to this minimal requirement,
the presheaf valued in the category of sets and maps corresponding to the Nash equilibrium concept or the Pareto efficiency criterion
fails to be well-defined.
Instead,
the presheaf on the category of sets and relations unavoidably emerges to achieve well-definedness.
(3)
Even when relations serve as morphisms,
functoriality, specifically the preservation of composition,
is not automatically guaranteed.
Therefore,
the explicit conditions for the presheaf to be well-defined are established.
(4)
The presheaf corresponding to the Nash equilibrium concept
requires a weaker condition to achieve well-definedness
than the one corresponding to the Pareto efficiency criterion:
the former requires order-reflecting maps between sets of strategy profiles,
whereas the latter requires order-embedding maps.

\cite{Lapitsky1999} defined a category of strategic games and a presheaf on it valued in the category of sets and \emph{maps} corresponding to the Nash equilibrium concept.
In contrast,
we show that
our presheaf valued in the category of sets and \emph{maps} corresponding to the Nash equilibrium concept are not well-defined.
Alternatively,
we show that
our presheaf valued in the category of sets and \emph{relations} is well-defined under a specific condition.
The results imply that
our minimal and natural definition of the category of strategic games,
where maps between sets of strategy profiles are merely required to be order-preserving, order-reflecting or order-embedding without any further constraints such as surjectivity,
requires the category of sets and \emph{relations} for the presheaf corresponding to the Nash equilibrium concept.
Further,
we consider the presheaf corresponding to the Pareto efficiency criterion.

\cite{Vannucci2024}, \cite{Tohme2025} and \cite{Jimenez2014} (resp. \cite{Streufert2021})
defined a category of strategic games (resp. extensive form games) in which morphisms include \emph{order-preserving} maps between sets of outcomes with respect to each player's preference.
\cite{Vannucci2024}
formulated an equilibrium concept as a functor from the \emph{discrete} category of strategic games to the category of sets.
\cite{Tohme2025}
formulated the Nash equilibrium concept as a map sending each game to the game whose set of outcomes is the set of Nash equilibrium outcomes in the original game,
which is \emph{not a functor}.
They defined a subcategory in which morphisms preserve Nash equilibria
and
showed that
if
the maps between sets of players in the morphisms are \emph{surjective}
and
those between sets of outcomes are \emph{surjective and order-embedding},
the morphisms preserve Nash equilibria.
\cite{Jimenez2014}
focused on strategic games with \emph{$n$ players and $n$ strategies of each player}
and
required the order-preserving maps to be \emph{bijective}.
They
showed that the morphisms preserve Nash equilibria.
\cite{Streufert2021} showed that
\emph{isomorphisms} preserve Nash equilibria and subgame perfect equilibria.
In this paper,
we show that
maps between sets of strategy profiles must be \emph{order-reflecting or order-embedding} for the \emph{presheaf} corresponding to the Nash equilibrium concept to be well-defined.
We \emph{do not impose}
any further constraints such as surjectivity on these maps
or
any constraints on strategic games.

\cite{Hedges2016} and \cite{Ghani2018} defined a category of games in which morphisms are \emph{equivalence classes of games} (known as open games)
to compose games.
In our paper,
the morphisms are not equivalence classes of games but \emph{pairs of a map between sets of players and a map between sets of strategy profiles}.

The remainder of this paper is organized as follows.
Section \ref{sec:definitions}
defines some set-theoretic and category-theoretic concepts.
Section \ref{sec:strategic_form_games}
defines a category of strategic games
and
shows that
it is well-defined.
Section \ref{sec:presheaves}
provides conditions
for a presheaf on the category of strategic games valued in a category of sets to be well-defined.
Section \ref{sec:nash_equilibrium_concept}
defines a presheaf corresponding to the Nash equilibrium concept
and
examines whether
it is well-defined.
Section \ref{sec:pareto_efficiency}
defines a presheaf corresponding to the Pareto efficiency criterion
and
examines whether
it is well-defined.
Section \ref{sec:conclusion} concludes the paper.
A proof sketch is provided immediately after each proposition to build intuition.
Full proofs of all propositions are provided in the Appendix.

\section{Definitions}\label{sec:definitions}

We define
set-theoretic concepts,
category-theoretic concepts
and
categories of sets.
These are standard definitions.

\subsection{Relations, maps and families}\label{subsec:sets}

For any $n \in \mathbb N$, any $n$-tuple $a$ and any $m \in \mathbb N$ with $m \leq n$,
let $a_m$ be the $m$th entry in $a$.

A \emph{relation} is a triple $\rnd{X,Y,G}$ of sets such that
$G \subset X \times Y$.
For any relation $R$,
let
$\dom\rnd{R} \coloneqq R_1$
(the \emph{domain of $R$}),
$\cod\rnd{R} \coloneqq R_2$
(the \emph{codomain of $R$})
and
$\graph\rnd{R} \coloneqq R_3$
(the \emph{graph of $R$}).
For
any relation $R$,
any $x \in \dom\rnd{R}$
and
any $y \in \cod\rnd{R}$,
write $x \mathrel{R} y$ to mean $\rnd{x,y} \in \graph\rnd{R}$.
For any sets $X$ and $Y$,
a \emph{relation between $X$ and $Y$} is a relation $R$ such that
$\dom\rnd{R} = X$ and $\cod\rnd{R} = Y$.
A \emph{left-total relation} (resp. \emph{right-total relation}) is a relation $R$ such that
for any $x \in \dom\rnd{R}$ (resp. $y \in \cod\rnd{R}$),
for some $y \in \cod\rnd{R}$ (resp. $x \in \dom\rnd{R}$),
$x \mathrel{R} y$.
A \emph{right-unique relation} is a relation
for any $x \in \dom\rnd{R}$,
for any $y,y' \in \cod\rnd{R}$,
if $x \mathrel{R} y$, and $x \mathrel{R} y'$,
then $y = y'$.

For
any relation $R$,
any $X \subset \dom\rnd{R}$
and $Y \subset \cod\rnd{R}$,
let
\begin{align*}
 R|_X^Y \coloneqq \rnd{X,Y,\graph\rnd{R} \cap \rnd{X \times Y}}
\end{align*}
(the \emph{restriction of $R$ to $X$ and $Y$}).
For any relations $R$ and $S$ such that $\cod\rnd{R} = \dom\rnd{S}$,
let
\begin{align*}
 S R
 \coloneqq
 \rnd
 {
  X,
  Z,
  \set
  {
   \rnd{x,z} \in X \times Z
  }
  {
   \exists y \in Y
   \rnd{\rnd{x \mathrel{R} y} \wedge \rnd{y \mathrel{S} z}}
  }
 },
\end{align*}
where
$X = \dom\rnd{R}$,
$Y = \cod\rnd{R} = \dom\rnd{S}$,
and
$Z = \cod\rnd{S}$
(the \emph{composite relation of $R$ and $S$}).
For any relation $X$,
let
\begin{align*}
 R^{-1}
 \coloneqq
 \rnd
 {
  Y,
  X,
  \set
  {\rnd{y,x} \in Y \times X}
  {x \mathrel{R} y}
 },
\end{align*}
where
$X = \dom\rnd{R}$,
and
$Y = \cod\rnd{R}$
(the \emph{inverse relation of $R$}).
An \emph{endorelation} is a relation $R$ such that
$\dom\rnd{R} = \cod\rnd{R}$.
For any set $X$,
let
\begin{align*}
 \id_X
 \coloneqq
 \rnd
 {
  X,
  X,
  \set
  {
   \rnd{x,y} \in X^2
  }
  {
   x = y
  }
 }
\end{align*}
(the \emph{identity relation on $X$}).

A \emph{map} is a left-total and right-unique relation.
For
any map $f$
and
any $x \in \dom\rnd{f}$,
let $f\rnd{x}$ be the element in $\cod\rnd{f}$ such that
$\rnd{x,f\rnd{x}} \in \graph\rnd{f}$
(the \emph{image of $x$ under $f$}).
For any map $f$
and
any $X \subset \dom\rnd{f}$,
let
\begin{align*}
 f\rnd{X}
 \coloneqq
 \set{y \in \cod\rnd{f}}{\exists x \in X \rnd{f\rnd{x} = y}}
\end{align*}
(the \emph{image of $X$ under $f$}).
For any map $f$
and
any $Y \subset \cod\rnd{f}$,
let
\begin{align*}
 f^{-1}\rnd{Y}
 \coloneqq
 \set{x \in \dom\rnd{f}}{\exists y \in Y \rnd{f\rnd{x} = y}}
\end{align*}
(the \emph{preimage of $Y$ under $f$}).
For any map $f$
and
any $y \in \cod\rnd{f}$,
let $f^{-1}\rnd{y} \coloneqq f^{-1}\rnd{\crly{y}}$
(the \emph{preimage of $y$ under $f$}).
For any sets $X$ and $Y$,
a \emph{map from $X$ to $Y$} is a map $f$ such that
$\dom\rnd{f} = X$ and $\cod\rnd{f} = Y$.
For any sets $X$ and $Y$,
let $\Map{X}{Y}$ be the set of maps from $X$ to $Y$.
For any maps $f$ and $g$ such that $\cod\rnd{f} = \dom\rnd{g}$,
\begin{align*}
 g f
 =
 \rnd
 {
  \dom\rnd{f},
  \cod\rnd{g},
  \set
  {\rnd{x,y} \in \dom\rnd{f} \times \cod\rnd{g}}
  {g\rnd{f\rnd{x}} = y}
 }
\end{align*}
(the \emph{composite map of $f$ and $g$}).
For any set $X$ and any $x \in X$,
$\id_X\rnd{x} = x$
(the \emph{identity map on $X$}).

A \emph{family} is a pair $\rnd{I,G}$ of sets such that
for some set $S$,
$\rnd{I,S,G}$ is a map from $I$ to $S$.\footnote
{
 Or equivalently,
 a family is a pair $\rnd{I,G}$ of sets such that
 for any $g \in G$,
 $g$ is a pair such that $g_1 \in I$,
 and
 for any $i \in I$,
 for some $g \in G$,
 $g_1 = i$,
 and
 for any $g,g'$,
 if $g_1 = g_1'$,
 then $g_2 = g_2'$.
 A family is a variant of a map,
 where the codomain is not specified.
}
For any family $f$,
let
$\dom\rnd{f} \coloneqq f_1$
(the \emph{domain of $f$}),
and
$\graph\rnd{f} \coloneqq f_2$
(the \emph{graph of $f$}).
For any family $f$ and any $i \in \dom\rnd{f}$,
let $f_i$ be the set such that
$\rnd{i,f_i} \in \graph\rnd{f}$
(the \emph{image of $i$ under $f$}).
For any set $I$,
a \emph{family indexed by $I$} is a family $f$ such that
$\dom\rnd{f} = I$.
For any family $f$ and $J \subset \dom\rnd{f}$,
let $f_J$ be the family such that
$\dom\rnd{f_J} = J$,
and
for any $i \in J$,
$\rnd{f_J}_i = f_i$
(the \emph{restriction of $f$ to $J$}).
For any family $f$ and any $i \in \dom\rnd{f}$,
let $f_{-i} \coloneqq f_{\dom\rnd{f} \setminus \crly{i}}$
(the restriction of $f$ to $\dom\rnd{f} \setminus \crly{i}$).
For any family $F$,
let $\prod F$ be the set of families $f$ such that
$\dom\rnd{f} = \dom\rnd{F}$,
and
for any $i \in \dom\rnd{f}$,
$f_i \in F_i$.
(the \emph{Cartesian product of $F$}).\footnote
{
 The standard notation for the Cartesian product of $F$ is $\prod_{i \in I} F_i$.
 Meanwhile,
 for simplicity,
 this paper's notation is $\prod F$.
}

\subsection{Categories}\label{subsec:categories}

A \emph{category $C$} consists of the following data:\footnote
{
 For the definition of a category,
 see \cite{Awodey2010}.
 Note that
 \cite{Awodey2010} uses ``arrow'' instead of ``morphism.''
}
\begin{enumerate}
 \item
 $\ob{C}$:
 a collection of \emph{objects}.
 \item
 $\mor{C}$:
 a collection of \emph{morphisms}.
 \item
 $\dom$:
 an operation sending each $f \in \mor{C}$ to a $\dom\rnd{f} \in \ob{C}$
 (the \emph{domain of $f$}).
 \item
 $\cod$:
 an operation sending each $f \in \mor{C}$ to a $\cod\rnd{f} \in \ob{C}$
 (the \emph{codomain of $f$}).
 \item
 $\circ$:
 an operation sending
 each pair of $f \in \mor{C}$ and $g \in \mor{C}$ such that
 $\cod\rnd{f} = \dom\rnd{g}$
 to a $g \circ f \in \mor{C}$ such that $\dom\rnd{g \circ f} = \dom\rnd{f}$, and $\cod\rnd{g \circ f} = \cod\rnd{g}$
 (the \emph{composite morphism of $f$ and $g$}).
 \item
 $\id$:
 an operation sending
 each $a \in \ob{C}$
 to an $\id_a \in \mor{C}$ such that
 $\dom\rnd{\id_a} = \cod\rnd{\id_a} = a$
 (the \emph{identity morphism on $a$}).
 \item
 For any $f,g,h \in \mor{C}$ such that
 $\cod\rnd{f} = \dom\rnd{g}$,
 and
 $\cod\rnd{g} = \dom\rnd{h}$,
 $h \circ \rnd{g \circ f} = \rnd{h \circ g} \circ f$
 (the associative law).
 \item
 For any $f \in \mor{C}$,
 $f \circ \id_{\dom\rnd{f}} = f$,
 and
 $\id_{\cod\rnd{f}} \circ f = f$
 (the unit law).
\end{enumerate}
For any $a,b \in \ob{C}$,
let $\hom{C}{a}{b}$ be the collection of $f \in \mor{C}$ such that
$\dom\rnd{f} = a$ and $\cod\rnd{f} = b$.
For any $f,g \in \mor{C}$,
say that $f$ and $g$ are \emph{composable}
if and only if $\cod\rnd{f} = \dom\rnd{g}$.
For any composable $f,g \in \mor{C}$,
write $g f$ instead of $g \circ f$.

For any category $C$,
let $\op{C}$ be the category as follows (the \emph{opposite category of $C$}):
\begin{enumerate}
 \item
 $\ob{\op{C}} = \ob{C}$.
 \item
 For any $a,b \in \ob{\op{C}}$,
 $\hom{\op{C}}{a}{b} = \hom{C}{b}{a}$.
 \item
 For any composable $f,g \in \mor{\op{C}}$,
 $g f$ in $\op{C}$ is equal to $f g$ in $C$.
 \item
 For any $a \in \ob{\op{C}}$,
 $\id_a$ in $\op{C}$ is equal to $\id_a$ in $C$.
\end{enumerate}

For any categories $C$ and $D$,
a \emph{functor from $C$ to $D$} is an operation as follows:
\begin{enumerate}
 \item
 $F$ sends each $a \in \ob{C}$ to an $F\rnd {a} \in \ob{D}$.
 \item
 For any $a,b \in \ob{C}$,
 $F$ sends each $f \in \hom{C}{a}{b}$ to an $F\rnd{f} \in \hom{D}{F\rnd{a}}{F\rnd{b}}$.
 \item
 For any composable $f,g \in \mor{C}$,
 $F\rnd{g f} = F\rnd{g} F\rnd{f}$
 (preservation of composition).
 \item
 For any $a \in \ob{C}$,
 $F\rnd{\id_a} = \id_{F\rnd{a}}$
 (preservation of identity).
\end{enumerate}

For any categories $C$ and $D$,
a \emph{contravariant functor from $C$ to $D$} is a functor from $\op{C}$ to $D$.
For any categories $C$ and $S$,
an \emph{$S$-valued presheaf on $C$} is a contravariant functor from $C$ to $S$.

\subsection{Category of sets}\label{subsec:rel_and_set}

Let $\Rel$ be the category as follows (the \emph{category of sets and relations}):
\begin{enumerate}
 \item
 $\ob{\Rel}$ is the collection of sets.
 \item
 For any $S,T \in \ob{\Rel}$,
 $\hom{\Rel}{X}{Y}$ is the collection of relation $R$ between $S$ and $T$.
 \item
 For any composable $R,S \in \mor{\Rel}$,
 the composite morphism of $R$ and $S$ is the composite relation of $R$ and $S$.
 \item
 For any $X \in \ob{\Rel}$,
 the identity morphism on $X$ is the identity relation on $X$.
\end{enumerate}

Let $\Set$ be the category as follows (the \emph{category of sets and maps}):
\begin{enumerate}
 \item
 $\ob{\Set}$ is the collection of sets.
 \item
 For any $S,T \in \ob{\Set}$,
 $\hom{\Set}{S}{T} = \Map{S}{T}$.
 \item
 For
 any composable $f,g \in \mor{\Set}$,
 the composite morphism of $f$ and $g$ is the composite map of $f$ and $g$.
 \item
 For any $X \in \ob{\Set}$,
 the identity morphisms on $X$ is the identity map on $X$.
\end{enumerate}

\section{Categories of strategic games}\label{sec:strategic_form_games}

We
define a category of strategic games.\footnote
{
 For the definition of a strategic game,
 see \cite{Osborne1994}.
}
We then show that
this presheaf is well-defined.

A \emph{preordered set} is a pair $\rnd{S,P}$ such that
$P$ is an endorelaton with $\dom\rnd{P} = \cod\rnd{P} = S$ such that
for any $s \in S$,
$s \mathrel{P} s$
(\emph{reflexivity}),
and
for any $s,s',s'' \in S$,
$\rnd{s \mathrel{P} s'} \wedge \rnd{s' \mathrel{P} s''} \rightarrow \rnd{s \mathrel{P} s''}$
(\emph{transitivity}).
A \emph{totally preordered set} is a preordered set $\rnd{S,P}$ such that
for any $s,s' \in S$,
$\rnd{s \mathrel{P} s'} \vee \rnd{s' \mathrel{P} s}$
(\emph{completeness}).\footnote
{
 In economics,
 the property satisfying both completeness and transitivity is called \emph{rationality}.
}
For any preordered sets $\rnd{S,P}$ and $\rnd{T,Q}$,
an \emph{order-preserving map} (resp. \emph{order-reflecting map}, \emph{order-embedding map}) from $\rnd{S,P}$ to $\rnd{T,Q}$ is a map $f \in \Map{S}{T}$ such that
for any $s,s' \in S$,
$\rnd{s \mathrel{P} s'} \rightarrow \rnd{f\rnd{s} \mathrel{Q} f\rnd{s'}}$
(resp. $\rnd{s \mathrel{P} s'} \leftarrow \rnd{f\rnd{s} \mathrel{Q} f\rnd{s'}}$,
$\rnd{s \mathrel{P} s'} \leftrightarrow \rnd{f\rnd{s} \mathrel{Q} f\rnd{s'}}$).

A \emph{strategic game} is a triple $\rnd{N,S,\succsim}$ such that
$\dom\rnd{S} = \dom\rnd{\succsim} = N$,
and
for any $i \in N$,
$\rnd{\prod S,\succsim_i}$ is a totally preordered set.
An element in $N$ is referred to as a \emph{player}.
For any $i \in N$,
an element of $S_i$ is referred to as a \emph{strategy} of player $i$.
For any $i \in N$,
$\succsim_i$ is referred to as player $i$'s \emph{preference relation}
($s \succsim_i s'$ means that
player $i$ weakly prefers $s$ to $s'$).
For any strategic game $G$,
let
$N^G \coloneqq G_1$,
$S^G \coloneqq G_2$
and
$\succsim^G \coloneqq G_3$.
For any strategic game $G$ and any $i \in N^G$,
let
\begin{align*}
 {\succ_i^G}&
 \coloneqq
 \rnd
 {
  \prod S^G,
  \prod S^G,
  \graph\rnd{\succsim_i^G} \setminus \graph\rnd{\rnd{\succsim_i^G}^{-1}}
 }\\
 {\sim_i^G}&
 \coloneqq
 \rnd
 {
  \prod S^G,
  \prod S^G,
  \graph\rnd{\succsim_i^G} \cap \graph\rnd{\rnd{\succsim_i^G}^{-1}}
 }\\
 {\prec_i^G}&
 \coloneqq
 \rnd
 {
  \prod S^G,
  \prod S^G,
  \graph{\rnd{\succsim_i^G}^{-1}} \setminus \graph\rnd{\succsim_i^G}
 }
\end{align*}
($s \succ_i^G s'$ (resp. $s \sim_i^G s'$, $s \prec_i^G s'$) means that
player $i$ prefers $s$ to $s'$ (resp. is indifferent between $s$ and $s'$, prefers $s'$ to $s$)).
For any strategic game $G$ and $M \subset N^G$,
let
\begin{align*}
 \succsim_M^G
 \coloneqq
 \rnd
 {
  \prod S^G,
  \prod S^G,
  \bigcap_{i \in M} \graph\rnd{\succsim_i^G}
 }
\end{align*}
($s \succsim_M^G s'$ means that
all players in $M$ weakly prefer $s$ to $s'$).

Throughout the remainder of this paper,
fix $\arrow$ to denote any of
$\rightarrow$, $\leftarrow$ or $\leftrightarrow$,
i.e.,
implication, converse implication or equivalence.

Let $\Gam$ be the category as follows:
\begin{enumerate}
 \item
 $\ob{\Gam}$ is the collection of strategic games.
 \item\label{it:morphism}
 For any $G,H \in \ob{\Gam}$,
 $\hom{\Gam}{G}{H}$ is the set of $\rnd{f_1,f_2} \in \Map{N^G}{N^H} \times \Map{\prod S^G}{\prod S^H}$ such that
  for any $j \in N^H$ and any $s,s' \in \prod S^G$,
 \begin{align}
  \rnd{s_{f_1^{-1}\rnd{j}} = s_{f_1^{-1}\rnd{j}}'}&
  \rightarrow
  \rnd{f_2\rnd{s}_j = f_2\rnd{s'}_j}\label{eq:preservation_of_strategy_tuples}\\
  \rnd{s \succsim_{f_1^{-1}\rnd{j}}^G s'}&
  \arrow
  \rnd{f_2\rnd{s} \succsim_j^H f_2\rnd{s'}}\label{eq:preservation_of_complete_preorders}
 \end{align}
 (if any player which $f_1$ sends to $j$ takes the same strategy in $s$ and $s'$,
 $j$ takes the same strategy in strategy profiles to which $f_2$ sends $s$ and $s'$;
 if $\arrow$ is $\rightarrow$ (resp. $\leftarrow$, $\leftrightarrow$), $f_2$ is an order-preserving map (resp. order-reflecting map, order-embedding map) from $\rnd{\prod S^G,\succsim_{f_1^{-1}\rnd{j}}^G}$ to $\rnd{\prod S^H,\succsim_j^H}$).
 \item
 For
 any composable $f,g \in \mor{\Gam}$,
 $g f = \rnd{g_1 f_1,g_2 f_2}$.\footnote
 {
  $g_1 f_1$ is the composite map of $f_1$ and $g_1$,
  and
  $g_2 f_2$ is the composite map of $f_2$ and $g_2$.
 }
 \item
 For any $G \in \ob{\Gam}$,
 $\id_{G} = \rnd{\id_{N^G},\id_{\prod S^G}}$.\footnote
 {
  $\id_{N^G}$ is the identity map on $N^G$,
  and
  $\id_{\prod S^G}$ is the identity map on $\prod S^G$.
 }
\end{enumerate}

Proposition \ref{prop:gam} shows that
$\Gam$ is well-defined.
\begin{proposition}\label{prop:gam}
 $\Gam$ is well-defined.
\end{proposition}
\begin{proofsketch}
 The issue is to show that
 composite morphisms satisfy (\ref{eq:preservation_of_complete_preorders}).
 This statement is guaranteed by transitivity of $\arrow$.
\end{proofsketch}

\section{Presheaves on category of strategic games}\label{sec:presheaves}

We define
an $\Rel$-valued or $\Set$-valued presheaf on $\Gam$.
We then present conditions for this presheaf to be well-defined.

An $\Rel$-valued or $\Set$-valued presheaf on $\Gam$
can represent a solution concept or a welfare criterion in game theory.
A solution concept, such as the Nash equilibrium concept, sends each strategic game to a specific set of strategy profiles realized by players' decision making.
A welfare criterion, such as the Pareto efficiency criterion, sends each strategic game to a specific set of strategy profiles desirable for the society of players.
Thus,
we can represent such a concept as a presheaf on $\Gam$ valued in a category of sets.
We employ $\Rel$ and $\Set$ as categories of sets.

Throughout the remainder of this paper,
fix $C$ to denote any of $\Rel$ or $\Set$.

Let $F$ be a $C$-valued presheaf on $\Gam$ such that
for any $G \in \ob{\Gam}$,
$F\rnd{G} \subset \prod S^G$,
and
for any $G,H \in \ob{\Gam}$ and any $f \in \hom{\Gam}{G}{H}$,
$F\rnd{f} = f_2^{-1}|_{F\rnd{H}}^{F\rnd{G}}$.\footnote
{
 $
  F\rnd{f}
  = \rnd
    {
     F\rnd{H},
     F\rnd{G},
     \set{\rnd{t,s} \in F\rnd{H} \times F\rnd{G}}{f_2\rnd{s} = t}
    }
  $.
}
$F\rnd{G}$ is a set of strategy profiles in $G$,
and
$F\rnd{f}$ is the restriction of the inverse of $f_2$ to $F\rnd{G}$ and $F\rnd{H}$.
$F\rnd{f}$ is a relation but not necessarily a map.

Condition \ref{con:equivalence_rel} states that
if there exists $s \in F\rnd{G}$ and $u \in F\rnd{I}$ such that
$f_2\rnd{s} = t$ and $g_2\rnd{t} = u$,
then $t \in F\rnd{H}$.
\begin{condition}\label{con:equivalence_rel}
 For
 any $G,H,I \in \ob{\Gam}$,
 any $f \in \hom{\Gam}{G}{H}$
 and
 any $g \in \hom{\Gam}{H}{I}$,
 $
  f_2\rnd{F\rnd{G}} \cap g_2^{-1}\rnd{F\rnd{I}}
  \subset
  F\rnd{H}
 $.
\end{condition}

Condition \ref{con:sufficiency_rel} states that
if there exists $t \in F\rnd{H}$ such that
$f_2\rnd{s} = t$,
then $s \in F\rnd{G}$.
\begin{condition}\label{con:sufficiency_rel}
 For
 any $G,H \in \ob{\Gam}$
 and
 any $f \in \hom{\Gam}{G}{H}$,
 $f_2^{-1}\rnd{F\rnd{H}} \subset F\rnd{G}$.
\end{condition}

Condition \ref{con:necessity_set} states that
if $t \in F\rnd{H}$,
then there exists $s \in F\rnd{G}$ such that
$f_2\rnd{s} = t$.
\begin{condition}\label{con:necessity_set}
 For
 any $G,H \in \ob{\Gam}$
 and
 any $f \in \hom{\Gam}{G}{H}$,
 $F\rnd{H} \subset f_2\rnd{F\rnd{G}}$.
\end{condition}

Proposition \ref{prop:conditions_rel} states that
Condition \ref{con:sufficiency_rel} implies Condition \ref{con:equivalence_rel}.
\begin{proposition}\label{prop:conditions_rel}
 If Condition \ref{con:sufficiency_rel} holds,
 then Condition \ref{con:equivalence_rel} holds.
\end{proposition}
\begin{proofsketch}
 By Condition \ref{con:sufficiency_rel},
 $g_2^{-1}\rnd{F\rnd{I}} \subset F\rnd{H}$ for $H$, $I$ and $g$ in Condition \ref{con:equivalence_rel}.
 Thus,
 Condition \ref{con:equivalence_rel} holds.
\end{proofsketch}

Proposition \ref{prop:presheaf_rel} presents an equivalent condition
for $F$ with $C = \Rel$ to be well-defined.
\begin{proposition}\label{prop:presheaf_rel}
 Suppose that
 $C = \Rel$.
 Then,
 $F$ is well-defined
 if and only if Condition \ref{con:equivalence_rel} holds.
\end{proposition}
\begin{proofsketch}
 If there is a case in which
 $s \in F\rnd{G}$,
 $f_2\rnd{s} \notin F\rnd{H}$,
 and
 $\rnd{g_2 f_2}\rnd{s} \in F\rnd{I}$,
 then
 $\rnd{\rnd{g_2 f_2}\rnd{s},s} \in \graph\rnd{F\rnd{g f}}$,
 and
 $\rnd{\rnd{g_2 f_2}\rnd{s},s} \notin \graph\rnd{F\rnd{f} F\rnd{g}}$;
 thus,
 $F$ is not well-defined.
 Condition \ref{con:equivalence_rel} excludes such a case.
\end{proofsketch}

Proposition \ref{prop:presheaf_set} presents a necessary condition
for $F$ with $C = \Set$ to be well-defined.
\begin{proposition}\label{prop:presheaf_set}
 Suppose that
 $C = \Set$.
 Then,
 $F$ is well-defined
 only if Condition \ref{con:necessity_set} holds.
\end{proposition}
\begin{proofsketch}
 Suppose that
 $F$ is well-defined.
 Then,
 $F\rnd{f} = f_2^{-1}|_{F\rnd{H}}^{F\rnd{G}}$ is a map.
 Thus,
 if $t \in F\rnd{H}$,
 for $s \coloneqq f_2^{-1}|_{F\rnd{H}}^{F\rnd{G}}\rnd{t}$,
 $s \in F\rnd{G}$,
 and $f_2\rnd{s} = t$;
 thus,
 $t \in f_2\rnd{F\rnd{G}}$.
\end{proofsketch}

\section{Presheaves corresponding to Nash equilibrium concept}\label{sec:nash_equilibrium_concept}

We define a $C$-valued presheaf on $\Gam$ corresponding to the Nash equilibrium concept.\footnote
{
 For the definition of a Nash equilibrium,
 see \cite{Osborne1994}.
}
We then present an equivalent condition for this presheaf with $C = \Rel$ to be well-defined.
We also show that
this presheaf with $C = \Set$ is not well-defined.

For any strategic game $G$,
a \emph{Nash equilibrium} in $G$ is $s \in \prod S^G$ such that
for
any $i \in N^G$
and
any $s' \in \prod S^G$ with $s_{-i}' = s_{-i}$,
$s \succsim_i^G s'$.
A Nash equilibrium is a strategy profile such that
no player is better off by deviating to their arbitrary strategy.

Let $\NE$ be the $C$-valued presheave on $\Gam$ such that
for
any $G,H \in \ob{\Gam}$
and
any $f \in \hom{\Gam}{G}{H}$,
\begin{align*}
 \NE\rnd{G}&
 =
 \set
 {
  s \in \prod S^G
 }
 {
  \forall i \in N^G
  \forall s' \in \prod S^G
  \rnd{\rnd{s_{-i} = s_{-i}'} \rightarrow \rnd{s \succsim_i^G s'}}
 }\\
 \NE\rnd{f}&
 =
 f_2^{-1}|_{\NE\rnd{H}}^{\NE\rnd{G}}.
\end{align*}
$\NE\rnd{G}$ is the set of Nash equilibria in $G$,
and
$\NE\rnd{f}$ is the restriction of the inverse of $f_2$ to the sets of Nash equilibria in $G$ and $H$.

Proposition \ref{prop:ne_rel} presents an equivalent condition
for $\NE$ with $C = \Rel$ to be well-defined.
\begin{proposition}\label{prop:ne_rel}
 Suppose that
 $C = \Rel$.
 Then,
 $\NE$ is well-defined
 if and only if
 ${\arrow} \in \crly{\leftarrow,\leftrightarrow}$.
\end{proposition}
\begin{proofsketch}
 Suppose that
 ${\arrow} \in \crly{\leftarrow,\leftrightarrow}$.
 Then,
 if no player is better off by unilaterally deviating from $f_2\rnd{s}$ in $H$,
 then no player is better off by unilaterally deviating from $s$ in $G$.
 Therefore,
 $\NE$ satisfies Condition \ref{con:sufficiency_rel}.

 Suppose that
 ${\arrow} = {\rightarrow}$.
 Consider
 strategic games $G$, $H$ and $I$ with one player $1$ such that
 $G$ has one strategy $0$,
 $H$ has two strategies $0$ and $1$ with $0 \prec_1^H 1$,
 and
 $I$ has two strategies $0$ and $1$ with $0 \sim_1^I 1$
 and
 morphisms $f \in \hom{\Gam}{G}{H}$ and $g \in \hom{\Gam}{H}{I}$ such that
 $f_2$ is an inclusion map,
 and
 $g_2$ is an identity map.
 Then,
 $0$ in $H$ is the image of a Nash equilibrium $0$ in $G$ under $f_2$ and a preimage of a Nash equilibrium $0$ in $I$ under $g_2$.
 However,
 $0$ is not a Nash equilibrium in $H$.
 Therefore,
 $\NE$ does not satisfy Condition \ref{con:equivalence_rel}.
\end{proofsketch}

Proposition \ref{prop:ne_set} shows that
$\NE$ with $C = \Set$ is not well-defined.
\begin{proposition}\label{prop:ne_set}
 Suppose that
 $C = \Set$.
 Then,
 $\NE$ is not well-defined.
\end{proposition}
\begin{proofsketch}
 Consider
 strategic games $G$ and $H$ with one player $1$ such that
 $G$ has one strategy $0$,
 and
 $H$ has two strategies $0$ and $1$ with $0 \prec_1^H 1$
 and
 a morphism $f \in \hom{\Gam}{G}{H}$ such that
 $f_2$ is an inclusion map.
 Then,
 $1$ is a Nash equilibrium in $H$.
 However,
 $1$ is not the image of a unique Nash equilibrium $0$ in $G$.
 Therefore,
 $\NE$ does not satisfy Condition \ref{con:necessity_set}.
\end{proofsketch}

\section{Presheaves corresponding to the Pareto efficiency criterion}\label{sec:pareto_efficiency}

We define a $C$-valued presheaf on $\Gam$ corresponding to the Pareto efficiency criterion.\footnote
{
 For the definition of Pareto efficiency,
 see \cite{Osborne1994}.
 Note that
 \cite{Osborne1994} use ``strongly Pareto efficient'' instead of ``Pareto efficient.''
}
We then present an equivalent condition for this presheaf with $C = \Rel$ to be well-defined.
We also show that
this presheaf with $C = \Set$ is not well-defined.

For any strategic game $G$ and any $s,s' \in \prod S^G$,
$s$ is \emph{Pareto superior} to $s'$
if and only if
for any $i \in N^G$,
$s' \succsim_i^G s$,
and
for some $i \in N^G$,
$s' \succ_i^G s$.
For any strategic game $G$ and any $s \in \prod S^G$,
$s$ is \emph{Pareto efficient} in $G$
if and only if
there exists no $s' \in \prod S^G$ that is Pareto superior to $s$.
A Pareto efficient strategy profile is a strategy profile such that
no joint deviation can make some player better off without making any player worse off.

Let $\PE$ be the $C$-valued presheaf on $\Gam$ such that
for
any $G,H \in \ob{\Gam}$
and
any $f \in \hom{\Gam}{G}{H}$,
\begin{align*}
 \PE\rnd{G}&
 =
 \set
 {
  s \in \prod S^G
 }
 {
  \neg
  \exists s' \in \prod S^G
  \rnd
  {
   \forall i \in N^G
   \rnd{s' \succsim_i^G s}
   \wedge
   \exists i \in N^G
   \rnd{s' \succ_i^G s}
  }
 }\\
 \PE\rnd{f}&
 =
 f_2^{-1}|_{\PE\rnd{H}}^{\PE\rnd{G}}.
\end{align*}
$\PE\rnd{G}$ is the set of Pareto efficient strategy profiles in $G$,
and
$\PE\rnd{f}$ is the restriction of the inverse of $f_2$ to the sets of Pareto efficient strategy profiles in $G$ and $H$.

Proposition \ref{prop:pe_rel} presents an equivalent condition
for $\PE$ with $C = \Rel$ to be well-defined.
\begin{proposition}\label{prop:pe_rel}
 Suppose that
 $C = \Rel$.
 Then,
 $\PE$ is well-defined
 if and only if
 ${\arrow} = {\leftrightarrow}$.
\end{proposition}
\begin{proofsketch}
 Suppose that
 ${\arrow} = {\leftrightarrow}$.
 Let
 $G$ and $H$ be strategic games
 and
 $f \in \hom{\Gam}{G}{H}$.
 Let $s$ be a strategy profile in $G$ such that
 $f_2\rnd{s} = t$ for some Pareto efficient strategy profile $t$ in $H$.
 Let $s'$ be a strategy profile in $G$ and $t' \coloneqq f_2\rnd{s'}$.
 $t$ is Pareto efficient in $H$.
 Thus,
 $t'$ is not Pareto superior to $t$ in $H$.
 Note that
 ${\arrow} = {\leftrightarrow}$.
 Then,
 $s'$ is not Pareto superior to $s$ in $G$.
 Thus,
 $s$ is Pareto efficient in $G$.
 Therefore,
 $\PE$ satisfies Condition \ref{con:sufficiency_rel}.

 Suppose that
 ${\arrow} = {\rightarrow}$.
 The argument is the same as in the case in which ${\arrow} = {\rightarrow}$ in the proof sketch of Proposition \ref{prop:ne_rel}.

 Suppose that
 ${\arrow} = {\leftarrow}$.
 Consider
 strategic games $G$, $H$ and $I$ with two players $1$ and $2$ and two strategy profiles $s$ and $s'$ such that
 $s \sim_1^G s'$,
 $s \sim_2^G s'$,
 $s \prec_1^H s'$,
 $s \sim_2^H s'$,
 $s \prec_1^I s'$,
 and
 $s \succ_2^I s'$
 and
 morphisms $f \in \hom{\Gam}{G}{H}$ and $g \in \hom{\Gam}{H}{I}$ such that $f_1$, $f_2$, $g_1$ and $g_2$ are identity maps. 
 $s$ in $H$ is the image of Pareto efficient $s$ in $G$ and a preimage of Pareto efficient $s$ in $I$.
 However,
 $s$ is not Pareto efficient in $H$.
 Therefore,
 $\PE$ does not satisfy Condition \ref{con:equivalence_rel}.
\end{proofsketch}

Proposition \ref{prop:pe_set} shows that
$\PE$ with $C = \Set$ is not well-defined.
\begin{proposition}\label{prop:pe_set}
 Suppose that
 $C = \Set$.
 Then,
 $\PE$ is not well-defined.
\end{proposition}
\begin{proofsketch}
 This proof sketch is the same as that of Proposition \ref{prop:ne_set}.
\end{proofsketch}

\section{Conclusion}\label{sec:conclusion}

We defined a category of strategic games, $\Gam$.
In $\Gam$,
a morphism is a pair of a map between sets of players and a map between sets of strategy profiles with specific properties.
The following three cases were considered for the maps between sets of strategy profiles:
they are order-preserving, order-reflecting or order-embedding with respect to each player's preference relation.
We showed that
$\Gam$ is well-defined.

We defined a presheaf on $\Gam$ valued in a category of sets.
This presheaf sends each strategic game to a set of strategy profiles.
The following two cases were considered for the morphisms of the category of sets:
they are relations or maps.
We presented an equivalent condition and a sufficient condition for the presheaf valued in the category of sets and relations to be well-defined.
We also presented a necessary condition for the presheaf valued in the category of sets and maps to be well-defined.

We defined presheaves, $\NE$ (resp. $\PE$), on $\Gam$ valued in a category of sets corresponding to the Nash equilibrium concept (resp. the Pareto efficiency criterion).
$\NE$ (resp. $\PE$) send each strategic game to the set of Nash equilibria (resp. Pareto efficient strategy profiles).
The following two cases were considered for the morphisms of the category of sets:
they are relations or maps.
We presented equivalent conditions for $\NE$ and $\PE$ valued in the category of sets and relations to be well-defined.
We also showed that
$\NE$ and $\PE$ valued in the category of sets and maps are not well-defined.
These results are summarized in Table \ref{tab:well-definedness_of_nash_equilibrium_concept_and_pareto_efficiency}.
``$\Rel$'' (resp. ``$\Set$'') indicates that
$\NE$ and $\PE$ are valued in the category of sets in which the morphisms are relations (resp. maps).
``$\rightarrow$'' (resp. ``$\leftarrow$,'' `$\leftrightarrow$'') indicates that
the maps between sets of strategy profiles in $\Gam$ are order-preserving (resp. order-reflecting, order-embedding).
``Well'' (resp. ``Not well'') indicates that
$\NE$ and $\PE$ are well-defined (resp. not well-defined).
\begin{table}[htbp]
 \begin{tabular}{c|c|c|c|}
  \multicolumn{1}{c}{} & \multicolumn{1}{c}{$\rightarrow$} & \multicolumn{1}{c}{$\leftarrow$} & \multicolumn{1}{c}{$\leftrightarrow$} \\ \cline{2-4}
  $\Rel$               & Not well                          & Well                             & Well                                  \\ \cline{2-4}
  $\Set$               & Not well                          & Not well                         & Not well                              \\ \cline{2-4}
 \end{tabular}
 \begin{tabular}{c|c|c|c|}
  \multicolumn{1}{c}{} & \multicolumn{1}{c}{$\rightarrow$} & \multicolumn{1}{c}{$\leftarrow$} & \multicolumn{1}{c}{$\leftrightarrow$} \\ \cline{2-4}
  $\Rel$               & Not well                          & Not well                         & Well                                  \\ \cline{2-4}
  $\Set$               & Not well                          & Not well                         & Not well                              \\ \cline{2-4}
 \end{tabular}
 \caption{Well-definedness of presheaves corresponding to the Nash equilibrium concept (left) and the Pareto efficiency criterion (right)}\label{tab:well-definedness_of_nash_equilibrium_concept_and_pareto_efficiency}
\end{table}
When we consider Nash equilibria,
we should use
$\NE$ valued in $\Rel$ on $\Gam$ with $\leftarrow$ or $\leftrightarrow$.
When we consider Pareto efficient strategy profiles,
we should use $\PE$ valued in $\Rel$ on $\Gam$ with $\leftrightarrow$.

\newpage
\appendix
\section*{Appendix}

\begin{proof}[Proof of Proposition \ref{prop:gam}]
 Show that
 composite morphisms satisfy \ref{it:morphism}. in the definition of $\Gam$.
 Let
 $G,H,I \in \ob{\Gam}$,
 $f \in \hom{\Gam}{G}{H}$
 and $g \in \hom{\Gam}{H}{I}$.
 Let $k \in N^I$ and $s,s' \in \prod S^G$.

 Show that
 $g f$ satisfies (\ref{eq:preservation_of_strategy_tuples}).
 Suppose that
 $s_{\rnd{g_1 f_1}^{-1}\rnd{k}} = s_{\rnd{g_1 f_1}^{-1}\rnd{k}}'$.
 Let $j \in g_1^{-1}\rnd{k}$.
 Let $i \in f_1^{-1}\rnd{j}$.
 Then,
 $f_1\rnd{i} = j$.
 Thus,
 $\rnd{g_1 f_1}\rnd{i} = g_1\rnd{f_1\rnd{i}}= g_1\rnd{j} = k$.
 Thus,
 $i \in \rnd{g_1 f_1}^{-1}\rnd{k}$.
 Thus,
 $s_i = s_i'$.
 Thus,
 $s_{f_1^{-1}\rnd{j}} = s_{f_1^{-1}\rnd{j}}'$.
 Thus,
 $f_2\rnd{s}_j = f_2\rnd{s'}_j$.
 Thus, 
 $f_2\rnd{s}_{g_1^{-1}\rnd{k}} = f_2\rnd{s'}_{g_1^{-1}\rnd{k}}$.
 Thus,
 $g_2 \rnd{f_2\rnd{s}}_k = g_2\rnd{f_2\rnd{s'}}_k$.
 Thus,
 $\rnd{g_2 f_2}\rnd{s}_k = \rnd{g_2 f_2}\rnd{s}_k$.

 Show that
 $g f$ satisfy (\ref{eq:preservation_of_complete_preorders}).
 \begin{align*}
  \rnd{s \succsim_{\rnd{g_1 f_1}^{-1}\rnd{k}}^G s'}&
  \leftrightarrow
  \forall i \in \rnd{g_1 f_1}^{-1}\rnd{k} \rnd{s \succsim_i^G s'}\\&
  \leftrightarrow
  \forall j \in g_1^{-1}\rnd{k} \forall i \in f_1^{-1}\rnd{j} \rnd{s \succsim_i^G s'}\\&
  \leftrightarrow
  \forall j \in g_1^{-1}\rnd{k} \rnd{s \succsim_{f_1^{-1}\rnd{j}}^G s'}\\&
  \arrow
  \forall j \in g_1^{-1}\rnd{k} \rnd{f_2\rnd{s} \succsim_j^H f_2\rnd{s'}}\\&
  \leftrightarrow
  \rnd{f_2\rnd{s} \succsim_{g_1^{-1}\rnd{k}}^H f_2\rnd{s'}}\\&
  \arrow
  \rnd{g_2\rnd{f_2\rnd{s}} \succsim_k^I g_2\rnd{f_2\rnd{s'}}}\\&
  \leftrightarrow
  \rnd{\rnd{g_2 f_2}\rnd{s} \succsim_k^I \rnd{g_2 f_2}\rnd{s'}}.
 \end{align*}
 Note that
 $\arrow$ is transitive.
 Then,
 $
  \rnd{s \succsim_{\rnd{g_1 f_1}^{-1}\rnd{k}}^G s'}
  \arrow
  \rnd{\rnd{g_2 f_2}\rnd{s} \succsim_k^I \rnd{g_2 f_2}\rnd{s'}}
 $.

 Obviously,
 identity morphisms satisfy \ref{it:morphism}. in the definition of $\Gam$.

 The associative law (resp. the unit law) follows the associative law (resp. the unit law) of maps.
\end{proof}

\begin{proof}[Proof of Proposition \ref{prop:conditions_rel}]
 Suppose that
 Condition \ref{con:sufficiency_rel} holds.
 Let
 $G,H,I \in \ob{\Gam}$,
 $f \in \hom{\Gam}{G}{H}$
 and
 $g \in \hom{\Gam}{H}{I}$.
 By Condition \ref{con:sufficiency_rel},
 $g_2^{-1}\rnd{F\rnd{I}} \subset F\rnd{H}$.
 Thus,
 $f_2\rnd{F\rnd{G}} \cap g_2^{-1}\rnd{F\rnd{I}} \subset F\rnd{H}$.
 Thus,
 Condition \ref{con:equivalence_rel} holds.
\end{proof}

\begin{proof}[Proof of Proposition \ref{prop:presheaf_rel}]
 \subproof{``if'' part}
 Suppose that
 Condition \ref{con:equivalence_rel} holds.
 Obviously,
 $F$ preserves identity.
 Show that
 $F$ preserves composition.
 Let
 $G,H,I \in \ob{\Gam}$,
 $f \in \hom{\Gam}{G}{H}$
 and
 $g \in \hom{\Gam}{H}{I}$.
 $\dom\rnd{F\rnd{g f}} = F\rnd{I} = \dom\rnd{F\rnd{f} F\rnd{g}}$,
 and
 $\cod\rnd{F\rnd{g f}} = F\rnd{G} = \cod\rnd{F\rnd{f} F\rnd{g}}$.
 Let $\rnd{u,s} \in \graph\rnd{F\rnd{g f}}$.
 Let $t = f_2\rnd{s}$.
 Then,
 $u = g_2\rnd{f_2\rnd{s}} = g_2\rnd{t}$.
 Note that
 $s \in F\rnd{G}$ and $u \in F\rnd{I}$.
 Then,
 $t \in f_2\rnd{F\rnd{G}} \cap g_2^{-1}\rnd{F\rnd{I}}$.
 Thus,
 by Condition \ref{con:equivalence_rel},
 $t \in F\rnd{H}$.
 Thus,
 $\rnd{u,s} \in \graph\rnd{F\rnd{s} F\rnd{g}}$.
 Obviously,
 $
  \graph\rnd{F\rnd{f} F\rnd{g}}
  \subset
  \graph\rnd{F\rnd{g f}}
 $.
 Thus,
 $
  \graph\rnd{F\rnd{g f}}
  =
  \graph\rnd{F\rnd{f} F\rnd{g}}
 $.
 Thus,
 $F\rnd{g f} = F\rnd{f} F\rnd{g}$.

 \subproof{``only if'' part}
 Suppose that
 $F$ is well-defined.
 Let
 $G,H,I \in \ob{\Gam}$,
 $f \in \hom{\Gam}{G}{H}$
 and
 $g \in \hom{\Gam}{H}{I}$.
 Let $t \in f_2\rnd{F\rnd{G}} \cap g_2^{-1}\rnd{F\rnd{I}}$.
 Then,
 there exist $s \in F\rnd{G}$ and $u \in F\rnd{I}$ such that
 $f_2\rnd{s} = t$
 and
 $g_2\rnd{t} = u$.
 Thus,
 $g_2\rnd{f_2\rnd{s}} = g_2\rnd{t} = u$.
 Thus,
 $\rnd{u,s} \in \graph\rnd{F\rnd{g f}}$.
 Note that
 by the supposition,
 $\graph\rnd{F\rnd{g f}} = \graph\rnd{F\rnd{f} F\rnd{g}}$.
 Then,
 $\rnd{u,s} \in \graph\rnd{F\rnd{f} F\rnd{g}}$.
 Thus,
 there exists $t' \in F\rnd{H}$ such that
 $f_2\rnd{s} = t'$
 and
 $g_2\rnd{t'} = u$.
 Note that
 $t = f_2\rnd{s} = t'$.
 Then,
 $t \in F\rnd{H}$.
 Thus,
 Condition \ref{con:equivalence_rel} holds.
\end{proof}

\begin{proof}[Proof of Proposition \ref{prop:presheaf_set}]
 Suppose that
 $F$ is well-defined.
 Let
 $G,H \in \ob{\Gam}$
 and
 $f \in \hom{\Gam}{G}{H}$.
 Let $t \in F\rnd{H}$.
 By the supposition,
 $F\rnd{f} = f_2^{-1}|_{F\rnd{H}}^{F\rnd{G}}$ is a map.
 Let $s = f_2^{-1}|_{F\rnd{H}}^{F\rnd{G}}\rnd{t}$.
 Then,
 $t = f_2\rnd{s}$.
 Note that
 $s \in F\rnd{G}$.
 Then,
 $t \in f_2\rnd{F\rnd{G}}$.
 Thus,
 Condition \ref{con:necessity_set} holds.
\end{proof}

\begin{lemma}\label{lem:well-definedness_rel_order-preserving}
 Suppose that
 $C = \Rel$
 and
 ${\arrow} = {\rightarrow}$.
 Let $F \in \crly{\NE,\PE}$.
 Then,
 $F$ is not well-defined.
\end{lemma}
\begin{proof}
 Let $G,H,I \in \ob{\Gam}$ such that
 $N^G = N^H = N^I = \crly{1}$,
 $S_1^G = \crly{0}$,
 $S_1^H = S_1^I = \crly{0,1}$,
 $1 \succ_1^H 0$,
 and
 $0 \sim_1^I 1$.
 Let
 $f \in \hom{\Gam}{G}{H}$ and $g \in \hom{\Gam}{H}{I}$ such that
 $f_2\rnd{0} = 0$,
 and
 $g_2 = \id_{\prod S^H}$.
 $0 \in F\rnd{G}$.
 Note that
 $f_2\rnd{0} = 0$.
 Then,
 $0 \in f_2\rnd{F\rnd{G}}$.
 $0 \in F\rnd{I}$.
 Note that
 $g_2\rnd{0} = 0$.
 Then,
 $0 \in g_2^{-1}\rnd{F\rnd{I}}$.
 Thus,
 $0 \in f_2\rnd{F\rnd{G}} \cap g_2^{-1}\rnd{F\rnd{I}}$.
 However,
 $0 \notin F\rnd{H}$.
 Thus,
 $F$ does not satisfy Condition \ref{con:equivalence_rel}.
 Thus,
 by Proposition \ref{prop:presheaf_rel},
 $F$ is not well-defined.
\end{proof}

\begin{lemma}\label{lem:well-definedness_set}
 Suppose that
 $C = \Set$.
 Let $F \in \crly{\NE,\PE}$.
 Then,
 $F$ is not well-defined.
\end{lemma}
\begin{proof}
 Let $G,H \in \ob{\Gam}$ such that
 $N^G = N^H = \crly{1}$,
 $S_1^G = \crly{0}$,
 $S_1^H = \crly{0,1}$,
 and
 $1 \succ_1^H 0$.
 Let $f \in \hom{\Gam}{G}{H}$ such that
 $f_2\rnd{0} = 0$.
 $1 \in F\rnd{H}$.
 Note that
 $f_2\rnd{F\rnd{G}} \subset f_2\rnd{\crly{0}} = \crly{0}$.
 Then,
 $1 \notin f_2\rnd{F\rnd{G}}$,
 Thus,
 $F$ does not satisfy Condition \ref{con:necessity_set}.
 Thus,
 by Proposition \ref{prop:presheaf_set},
 $F$ is not well-defined.
\end{proof}

\begin{proof}[Proof of Proposition \ref{prop:ne_rel}]
 \subproof{``if'' part}
 Suppose that
 ${\arrow} \in \crly{\leftarrow,\leftrightarrow}$.
 Let
 $G,H \in \ob{\Gam}$
 and
 $f \in \hom{\Gam}{G}{H}$.
 Let $s \in f_2^{-1}\rnd{\NE\rnd{H}}$.
 Let $i \in N$.
 Let $s' \in \prod S^G$ such that
 $s_{-i}' = s_{-i}$.
 Let $j = f_1\rnd{i}$.
 Then,
 for any $j' \in J \setminus \crly{j}$,
 $s_{f_1^{-1}\rnd{j'}} = s_{f_1^{-1}\rnd{j'}}'$,
 and
 thus,
 $f_2\rnd{s}_{j'} = f_2\rnd{s'}_{j'}$.
 Thus,
 $f_2\rnd{s}_{-j} = f_2\rnd{s'}_{-j}$.
 Note that
 $f_2\rnd{s} \in \NE\rnd{H}$.
 Then,
 $f_2\rnd{s} \succsim_{j}^H f_2\rnd{s'}$.
 Thus,
 $s \succsim_{f_1^{-1}\rnd{j}}^G s'$.
 Thus,
 $s \succsim_i^G s'$.
 Thus,
 $s \in \NE\rnd{G}$.
 Thus,
 $\NE$ satisfies Condition \ref{con:sufficiency_rel}.
 Thus,
 by Proposition \ref{prop:conditions_rel},
 $\NE$ satisfies Condition \ref{con:equivalence_rel}.
 Thus,
 by Proposition \ref{prop:presheaf_rel},
 $\NE$ is well-defined.

 \subproof{``only if'' part}
 Suppose that
 ${\arrow} = {\rightarrow}$.
 Then,
 by Lemma \ref{lem:well-definedness_rel_order-preserving},
 $\NE$ is not well-defined.
\end{proof}

\begin{proof}[Proof of Proposition \ref{prop:ne_set}]
 By Lemma \ref{lem:well-definedness_set},
 $\NE$ is not well-defined.
\end{proof}

\begin{proof}[Proof of Proposition \ref{prop:pe_rel}]
 \subproof{``if'' part}
 Suppose that
 ${\arrow} = {\leftrightarrow}$.
 Let
 $G,H \in \ob{\Gam}$,
 and $f \in \hom{\Gam}{G}{H}$.
 Let $s \in f_2^{-1}\rnd{\PE\rnd{H}}$.
 Then,
 there exists $t \in \PE\rnd{H}$ such that
 $f_2\rnd{s} = t$.
 Let $s' \in \prod S^G$.
 Let $t' = f_2\rnd{s'}$.
 Because $t$ is Pareto efficient in $H$,
 $t'$ is not Pareto superior to $t$ in $H$.
 Thus,
 for some $j \in N^H$,
 $t \succ_j^H t'$,
 or
 for any $j \in N^H$,
 $t \succsim_i^H t'$.
 \begin{enumerate}
  \item
  Consider the case in which
  for some $j \in N^H$,
  $t \succ_j^H t'$.
  Note that
  \begin{align*}
   \rnd{t \succ_j^H t'}&
   \leftrightarrow
   \neg \rnd{t' \succsim_j^H t}\\&
   \leftrightarrow
   \neg \rnd{s' \succsim_{f_1^{-1}\rnd{j}}^G s}\\&
   \leftrightarrow
   \neg \forall i \in f_1^{-1}\rnd{j} \rnd{s' \succsim_i^G s}\\&
   \leftrightarrow
   \exists i \in f^{-1}\rnd{j} \neg \rnd{s' \succsim_i^G s}\\&
   \leftrightarrow
   \exists i \in f^{-1}\rnd{j} \rnd{s \succ_i^G s'}.
  \end{align*}
  Then,
  for some $i \in N^G$,
  $s \succ_i^G s'$.
  \item
  Consider the case in which
  for any $j \in N^H$,
  $t \succsim_j^H t'$.
  Let $i \in N^G$.
  Let $j' = f_1\rnd{i}$.
  Then,
  $t \succsim_{j'}^H t'$.
  Thus,
  $s \succsim_{f_1^{-1}\rnd{j'}}^H s'$.
  Thus,
  $s \succsim_i^G s'$.
  Thus,
  for any $i \in N^G$,
  $s \succsim_i^G s'$.
 \end{enumerate}
 Thus,
 for some $i \in N^G$,
 $s \succ_i^G s'$,
 or
 for any $i \in N^G$,
 $s \succsim_i^G s'$.
 Thus,
 $s'$ is not Pareto superior to $s$ in $G$.
 Thus,
 $s \in \PE\rnd{G}$.
 Thus,
 $\PE$ satisfies Condition \ref{con:sufficiency_rel}.
 Thus,
 by Proposition \ref{prop:conditions_rel},
 $\PE$ satisfies Condition \ref{con:equivalence_rel}.
 Thus,
 by Proposition \ref{prop:presheaf_rel},
 $\PE$ is well-defined.

 \subproof{``only if'' part}
 If ${\arrow} = {\rightarrow}$,
 by Lemma \ref{lem:well-definedness_rel_order-preserving},
 $\PE$ is not well-defined.
 Suppose that
 ${\arrow} = {\leftarrow}$.
 Let $G,H,I \in \ob{\Gam}$ such that
 for any $K \in \crly{G,H,I}$,
 $N^K = \crly{1,2}$,
 $\prod S^K = \crly{s,s'}$ for some distinct $s$ and $s'$,
 and
 \begin{align*}
  \begin{array}{ccc}
   s \sim_1^G s', & s \prec_1^H s', & s \prec_1^I s', \\
   s \sim_2^G s', & s \sim_2^H s',  & s \succ_2^I s'.
  \end{array}
 \end{align*}
 Let $f \coloneqq \rnd{\id_{N^G},\id_{\prod S^G}}$ and $g \coloneqq \rnd{\id_{N^G},\id_{\prod S^G}}$.
 $s \in \PE\rnd{G}$.
 Note that
 $f_2\rnd{s} = s$.
 Then,
 $s \in f_2\rnd{\PE\rnd{G}}$.
 $s \in \PE\rnd{I}$.
 Note that
 $g_2\rnd{s} = s$.
 Then,
 $s \in g_2^{-1}\rnd{\PE\rnd{I}}$.
 Thus,
 $s \in f_2\rnd{\PE\rnd{G}} \cap g_2^{-1}\rnd{\PE\rnd{I}}$.
 However,
 $s \notin \PE\rnd{H}$.
 Thus,
 $\PE$ does not satisfy Condition \ref{con:equivalence_rel}.
 Thus,
 by Proposition \ref{prop:presheaf_rel},
 $\PE$ is not well-defined.
\end{proof}

\begin{proof}[Proof of Proposition \ref{prop:pe_set}]
 By Lemma \ref{lem:well-definedness_set},
 $\PE$ is not well-defined.
\end{proof}

\newpage
\bibliographystyle{plainnat}
\bibliography{Category_of_Games_1}

\begin{thebibliography}{9}
\providecommand{\natexlab}[1]{#1}
\providecommand{\url}[1]{\texttt{#1}}
\expandafter\ifx\csname urlstyle\endcsname\relax
  \providecommand{\doi}[1]{doi: #1}\else
  \providecommand{\doi}{doi: \begingroup \urlstyle{rm}\Url}\fi

\bibitem[Awodey(2010)]{Awodey2010}
S.~Awodey.
\newblock \emph{Category Theory, Second Edition}.
\newblock Oxford University Press, 2010.

\bibitem[Ghani et~al.(2018)Ghani, Hedges, Winschel, and Zahn]{Ghani2018}
N.~Ghani, J.~Hedges, V.~Winschel, and P.~Zahn.
\newblock Compositional game theory.
\newblock In \emph{LICS '18: Proceedings of the 33rd Annual ACM/IEEE Symposium
  on Logic in Computer Science}, 2018.

\bibitem[Hedges(2016)]{Hedges2016}
J.~Hedges.
\newblock \emph{Towards Compositional Game Theory}.
\newblock PhD thesis, Queen Mary University of London, 2016.

\bibitem[Jim\'enez(2014)]{Jimenez2014}
A.~Jim\'enez.
\newblock Game theory from the category theory point of view.
\newblock Working paper, 2014.

\bibitem[Lapitsky(1999)]{Lapitsky1999}
V.~Lapitsky.
\newblock On some categories of games and corresponding equilibria.
\newblock \emph{International Game Theory Review}, 1:\penalty0 169--185, 1999.

\bibitem[Osborne and Rubinstein(1994)]{Osborne1994}
M.~J. Osborne and A.~Rubinstein.
\newblock \emph{A Course in Game Theory}.
\newblock MIT Press, 1994.

\bibitem[Streufert(2021)]{Streufert2021}
P.~A. Streufert.
\newblock A category for extensive-form games.
\newblock arXiv:2105.11398, 2021.

\bibitem[Tohm\'e and Viglizzo(2025)]{Tohme2025}
F.~Tohm\'e and I.~Viglizzo.
\newblock A categorical representation of games.
\newblock arXiv:2309.15981, 2025.

\bibitem[Vannucci(2024)]{Vannucci2024}
S.~Vannucci.
\newblock Categories of games and chu spaces.
\newblock Technical Report 918, Department of Economics and Statistics,
  University of Siena, 2024.

\end{thebibliography}

\end{document}